\documentclass[a4paper, UKenglish]{lipics-v2016}

\usepackage[utf8]{inputenc}

\usepackage{newunicodechar}
\newunicodechar{α}{\ensuremath{\alpha}\xspace}
\newunicodechar{β}{\ensuremath{\beta}\xspace}
\newunicodechar{γ}{\ensuremath{\gamma}\xspace}
\newunicodechar{σ}{\ensuremath{\sigma}\xspace}
\newunicodechar{δ}{\ensuremath{\delta}\xspace}
\newunicodechar{τ}{\ensuremath{\tau}\xspace}
\newunicodechar{φ}{\ensuremath{\varphi}\xspace}
\newunicodechar{ε}{\ensuremath{\varepsilon}\xspace}
\newunicodechar{π}{\ensuremath{\pi}\xspace}
\newunicodechar{ψ}{\ensuremath{\psi}\xspace}
\newunicodechar{θ}{\ensuremath{\theta}\xspace}
\newunicodechar{λ}{\ensuremath{\lambda}\xspace}

\newunicodechar{Λ}{\ensuremath{\Lambda}\xspace}
\newunicodechar{Ψ}{\ensuremath{\Psi}\xspace}
\newunicodechar{Γ}{\ensuremath{\Gamma}\xspace}
\newunicodechar{Φ}{\ensuremath{\Phi}\xspace}
\newunicodechar{Π}{\ensuremath{\Pi}\xspace}
\newunicodechar{Δ}{\ensuremath{\Delta}\xspace}
\newunicodechar{Θ}{\ensuremath{\Theta}\xspace}

\newunicodechar{∃}{\ensuremath{\exists}\xspace}
\newunicodechar{∀}{\ensuremath{\forall}\xspace}
\newunicodechar{ℕ}{\ensuremath{\mathbb{N}}\xspace}
\newunicodechar{ℤ}{\ensuremath{\mathbb{Z}}\xspace}
\newunicodechar{→}{\ensuremath{\to}\xspace}
\newunicodechar{⇐}{\ensuremath{\Leftarrow}\xspace}
\newunicodechar{⇔}{\ensuremath{\Leftrightarrow}\xspace}
\newunicodechar{⇒}{\ensuremath{\Rightarrow}\xspace}
\newunicodechar{↦}{\ensuremath{\mapsto}\xspace}
\newunicodechar{•}{\ensuremath{\cdot}\xspace}

\newunicodechar{¬}{\ensuremath{\neg}\xspace}
\newunicodechar{∨}{\ensuremath{\lor}\xspace}
\newunicodechar{∧}{\ensuremath{\land}\xspace}
\newunicodechar{⊥}{\ensuremath{\bot}\xspace}

\usepackage{color}

\usepackage{xspace}
\usepackage{amsmath}
\usepackage{amssymb}
\usepackage{multicol}
\usepackage{braket}
\usepackage{mathtools}
\usepackage{enumerate}
\usepackage{tikz}
\usetikzlibrary{chains}
\usetikzlibrary{calc}

\makeatletter
\newcommand*{\defeq}{\mathrel{%
  \rlap{\raisebox{0.3ex}{$\m@th\cdot$}}%
  \raisebox{-0.3ex}{$\m@th\cdot$}}%
  =}\makeatother
\makeatletter
\newcommand*{\eqdef}{=\mathrel{%
  \raisebox{0.3ex}{$\m@th\cdot$}%
  \llap{\raisebox{-0.3ex}{$\m@th\cdot$}}}%
  }\makeatother
\makeatletter
\newcommand*{\defeqv}{\mathrel{%
  \rlap{\raisebox{0.3ex}{$\m@th\cdot$}}%
  \raisebox{-0.3ex}{$\m@th\cdot$}}%
  \Longleftrightarrow}\makeatother
\makeatletter
\newcommand*{\eqvdef}{\Longleftrightarrow\mathrel{%
  \raisebox{0.3ex}{$\m@th\cdot$}%
  \llap{\raisebox{-0.3ex}{$\m@th\cdot$}}}%
  }\makeatother

\newcommand{\win}{\classFont{Win}}

\newcommand\classFont[1]{\textnormal{#1}}

\newcommand{\NP}{\protect\ensuremath{\classFont{NP}}\xspace}
\newcommand{\cP}{\protect\ensuremath{{\classFont{\#P}}}\xspace}
\newcommand{\PP}{\protect\ensuremath{\classFont{PP}}\xspace}

\newcommand{\AC}{\protect\ensuremath{{\classFont{AC}^0}}\xspace}
\newcommand{\NC}{\protect\ensuremath{{\classFont{NC}^1}}\xspace}
\newcommand{\SAC}{\protect\ensuremath{{\classFont{SAC}^1}}\xspace}
\newcommand{\ACone}{\protect\ensuremath{{\classFont{AC}^1}}\xspace}
\newcommand{\cAC}{\protect\ensuremath{{\#\classFont{AC}^0}}\xspace}

\newcommand{\FO}{\protect\ensuremath{\classFont{FO}}\xspace}

\newcommand{\WinFO}{\protect\ensuremath{\classFont{\#Win-FO}}\xspace}

\newcommand{\struc}{\protect\ensuremath{\textnormal{STRUC}}}
\newcommand{\tString}{\protect\ensuremath{τ_{\textnormal{string}}}\xspace}
\newcommand{\tCirc}{\protect\ensuremath{τ_{\textnormal{circ}}}}

\newcommand{\tu}[1]{\overline{#1}}
\newcommand{\gpr}{\textnormal{GPR}\xspace}
\newcommand{\semi}{\textnormal{semi}}
\newcommand{\bounded}{\textnormal{bound}}

\newcommand{\ie}{i.e.\@\xspace}
\newcommand{\eg}{e.g.\@\xspace}

\newcommand{\ST}{such that\@\xspace}

\newcommand{\fa}{for all\@\xspace}
\newcommand{\stfa}{such that for all\@\xspace}

\title{Model-Theoretic Characterizations of Boolean and Arithmetic Circuit Classes of Small Depth}
\author[1]{Arnaud Durand}
\author[2]{Anselm Haak}
\author[2]{Heribert Vollmer}
\affil[1]{Université Paris Diderot, IMJ-PRG, CNRS UMR 7586, Case 7012, 75205 Paris cedex 13, France\\
  \texttt{durand@math.univ-paris-diderot.fr}}
\affil[2]{Theoretische Informatik, Leibniz Universität Hannover, Appelstraße, D-30167, Germany\\
\texttt{(haak|vollmer)@thi.uni-hannover.de}}
\authorrunning{A. Durand, A. Haak, H. Vollmer}
\titlerunning{A Model-Theoretic Characterization of Constant-Depth Arithmetic Circuits}
\Copyright{Arnaud Durand, Anselm Haak, Heribert Vollmer}
\subjclass{F.1.1 Models of Computation, F.1.3 Complexity Measures and Classes, F.4.1 Mathematical Logic}
\keywords{finite model theory, descriptive complexity, arithmetic circuits, counting classes}

\begin{document}
\maketitle


\begin{abstract}
In this paper we give a characterization of both Boolean and arithmetic circuit classes of logarithmic depth in the vein of descriptive complexity theory, i.e., the Boolean classes $\NC$, $\SAC$ and $\ACone$ as well as their arithmetic counterparts $\#\NC$, $\#\SAC$ and $\#\ACone$. We build on Immerman's characterization of constant-depth polynomial-size circuits by formulae of first-order logic, i.e., $\AC=\FO$, and augment the logical language with an operator for defining relations in an inductive way. Considering slight variations of the new operator, we obtain uniform characterizations of the three just mentioned Boolean classes. The arithmetic classes can then be characterized by functions counting winning strategies in semantic games for formulae characterizing languages in the corresponding Boolean class.
\end{abstract}

\section{Introduction}

The computational power of arithmetic circuits is of current focal interest in computational complexity theory, see the recent surveys \cite{Mah2014,KaSa2014} or the continuously updated collection of results at \cite{Sap2017}. A number of very powerful techniques to prove lower bounds for such circuits have been developed, however only for quite restricted classes. 

A long line of research in computational complexity is to characterize complexity classes in a model-theoretic way. Instead of constructing a computational device such as a Turing machine or a family of circuits \emph{deciding} a language $L$, a formula is built that \emph{defines} the property of those words in $L$. Best-known is probably Fagin's Theorem stating that languages in $\NP$ are exactly those that can be defined in existential second-order logic. More important for this paper is Immerman's theorem, in which the circuit class $\AC$ of all languages decidable by Boolean circuits of polynomial size and constant depth is addressed: Immerman showed that $\AC$ equals the class of languages definable by first-order formulae: $\AC=\FO$ \cite{Immerman89}. The rationale behind this area of \emph{descriptive complexity}, as it is called, is to characterize complexity classes in a model-theoretic way in order to better understand their structure, and to use logical methods in order to get new insights about the considered classes and, most prominently, to obtain lower bounds, see the monographs \cite{ImmermanBuch,Libkin04}. The famous lower bound for $\AC$, showing that the parity function cannot be computed by such circuit families \cite{FurstSS84}, was obtained independently by Ajtai \cite{Ajtai83} in a purely logical way.

For arithmetic circuit classes, only one descriptive complexity characterization is known to date. Generalizing in a sense Immerman's Theorem, it was shown very recently that the class $\cAC$ of those functions from binary words to natural numbers computable by polynomial-size constant-depth arithmetic circuits with plus and times gates is equal to the class of those functions computing winning strategies in semantic games for first-order formulae: $\cAC = \WinFO$ \cite{HaakV16}. A different way to view this result is to say that $\cAC$ is the class of functions counting Skolem functions for $\FO$-formulae.

Central for this result is a way of looking at arithmetic computation as a counting process: Say that a \emph{proof tree} of a Boolean circuit $C$ for a given input word $w$ is a minimal subtree (of the circuit unfold into a tree) witnessing that the circuit outputs $1$ on input $w$, and let $\#C(w)$ denote the number of such proof trees. It is folklore that $\cAC$ consists of those functions counting proof trees for $\AC$-circuits. To prove the mentioned result from \cite{HaakV16}, a formula has to be constructed whose number of winning strategies (or, number of Skolem functions) equals the number of proof trees of the original circuit. 

The aim of this paper is to generalize the theorem $\cAC = \WinFO$ to larger circuit classes, in particular the classes $\#\NC$, $\#\SAC$ and $\#\ACone$, defined by families of arithmetic circuits of polynomial size and logarithmic depth with bounded fan-in addition and multiplication gates (for $\#\NC$), unbounded fan-in addition and bounded fan-in multiplication gates ($\#\SAC$), and unbounded fan-in addition and multiplication gates ($\#\ACone$), see \cite{HeribertBuch}.
The mentioned equality between the value computed by an arithmetic circuit and the number of proof trees of the corresponding Boolean circuit does not only hold in the case of the class $\AC$ but is a general observation. Thus, a reasonable roadmap to obtain our generalization seems to study logical characterizations of the corresponding decision classes $\NC$, $\SAC$ and $\ACone$. Such characterizations can be found in the literature: $\NC$ can be characterized by an extension of first-order logic by so called monoidal quantifiers \cite{baimst90}, and similarly $\SAC$ by extending $\FO$ by groupoidal quantifiers \cite{LautemannMSV01}. 
However, for such logics with generalized quantifier the notion of winning strategy is not clear. 
Following a completely different approach, Immerman extended first-order logic by allowing repeated quantifier blocks and thus characterized $\ACone$ \cite{Immerman89}. Here it can be said that in Immerman's notation, $\#\ACone=\WinFO[\log]$, but this result cannot be transfered to the other log-depth classes $\NC$ and $\SAC$.
Hence we have to start by developing new logical characterizations for the Boolean classes $\NC$, $\SAC$ and $\ACone$. 

Inspiration comes from a result by Compton and Laflamme, characterizing $\NC$ by $\FO$ logic augmented with the RPR-operator allowing to define relations by a certain kind of linear recursion \cite{ComptonL90} (RPR stands for relational primitive recursion). This approach does not generalize to the classes $\SAC$ and $\ACone$, though. Also, the number of winning strategies does not seem to be related to the number of proof trees; so again, their approach is not suitable for our aim. Instead, we consider a new operator, called GPR (``guarded predicative recursion''), allowing to define relations by a certain kind of parallel recursion. We show that $\FO(\gpr)$, first-order logic augmented by GPR, characterizes $\ACone$, and that slight modifications of the GPR-operator lead to characterizations of $\NC$ and $\SAC$. In a second step, we show that these characterizations are in a sense ``close enough'' to the circuit world to mirror the process of counting proof trees by counting winning strategies in semantic games.  

Our paper is structured as follows. In the next section, we will give the necessary preliminaries from first order logic and circuit complexity including the respective counting mechanism. In Sect.~\ref{sect:GPR} we briefly recall the result by Compton and Laflamme and then introduce our inductive operator \gpr. To demonstrate suitability of our logical approach, we give an example of a formula defining an $\ACone$-complete problem. We then prove our main results: In Sect.~\ref{sect:dec} we characterize the Boolean classes $\NC$, $\SAC$ and $\ACone$ in a model-theoretic way by first-order logic with different forms of the GPR-operator. This is the technically most demanding part of our paper. We would like to stress that our proofs are completely different from the one for the mentioned result from Compton and Laflamme \cite{ComptonL90}. In Sect.~\ref{sect:count} we characterize the arithmetic classes $\#\NC$, $\#\SAC$ and $\#\ACone$ by counting winning-strategies in semantic games for the above logics. Finally, we conclude with a summary and some open problems.

\section{Preliminaries}

In this paper we will use first-order logic \FO with usual syntax and semantics, see, \eg, \cite{EbbFlumThomBuch}.
We consider finite σ-structures where σ is a finite vocabulary consisting of relation and constant symbols. For a structure $\mathcal{A}$, $\text{dom}(\mathcal{A})$ denotes its universe. We will always use structures with universe $\{0, 1, \dots, n-1\}$ for some $n \in ℕ \setminus \{0\}$. Furthermore, we will always assume that our structures contain the \emph{built-in relation} $\text{BIT}^2$, which is implicitly interpreted in the expected way: $\text{BIT}(i,j)$ is true, iff the $i$'th bit of the binary representation of $j$ is 1. When talking about structures with built-in relations, $\vDash$ includes the interpretation of the built-in relations in the intended way.

We assume the standard encoding of structures as binary strings (see, \eg, \cite{ImmermanBuch}): Relations are encoded row by row by listing their truth values as 0's and 1's. Constants are encoded by the binary representation of their value and thus a string of length $\lceil\log_2(n)\rceil$. A whole structure is encoded by the concatenation of the encodings of its relations and constants except for numerical predicates and constants: These are not encoded, because they are determined by the input length.

Since we want to talk about languages accepted by Boolean circuits, we will need the vocabulary
\[\tString = (\leq^2, S^1)\]
of binary strings. A binary string is represented as a structure over this vocabulary as follows: Let $w \in \{0,1\}^*$ with $|w| = n$. Then the structure representing this string is the structure with universe $\{0, \dots, n-1\}$, $\leq^2$ interpreted as the $\leq$-relation on ℕ restricted to the universe and $x \in S$, iff the $x$'th bit of $w$ is 1. The structure corresponding to string $w$ is denoted by $\mathcal{A}_w$. Also, by the above, $w$ is the encoding of structure $\mathcal{A}_w$.

We denote by \FO not only the set of first-order formulae, but also the complexity class of all languages definable in first-order logic with built-in BIT:
\begin{definition}\label{FOclass}
A language $L \subseteq \{0,1\}^*$ is in \FO if there is an \FO-formula φ over vocabulary $\tString \cup (\text{BIT}^2)$ \stfa $w \in \{0,1\}^*$:
\[w \in L ⇔ \mathcal{A}_w \vDash φ.\]
\end{definition}

We will also use relativized quantifiers. A relativization of a quantifier is a formula restricting the domain of elements considered for that quantifier. More precisely, we write
\[(∃x.φ) \; ψ\]
as a shorthand for $∃x (φ ∧ ψ)$ and, respectively,
\[(∀x.φ) \; ψ\]
as a shorthand for $∀x (φ → ψ) \equiv ∀x (\neg φ ∨ ψ)$.

Furthermore, we consider bounded variants of relativized quantifiers, that is, quantifiers where we only consider the maximal two elements meeting the condition expressed by the relativization. Notation: $∃_\text{b}$, $∀_\text{b}$. Formally, the semantics can be given in \FO as follows:
\[(∃_\text{b} x.φ(x)) \; ψ(x) \equiv \Big(∃x.\big(φ(x) ∧ ∀y∀z \ (y \neq z ∧ x<y ∧ x<z) → (¬φ(y) ∨ ¬φ(z))\big)\Big) \; ψ(x)\]
\[(∀_\text{b} x.φ(x)) \; ψ(x) \equiv \Big(∀x.\big(φ(x) ∧ ∀y∀z \ (y \neq z ∧ x<y ∧ x<z) → (¬φ(y) ∨ ¬φ(z))\big)\Big) \; ψ(x)\]

For the definition of uniform circuit families we will need \FO-interpretations, which are mappings between structures over different vocabularies.

\begin{definition}
Let σ, τ be vocabularies, $τ = (R_1^{a_1}, \dots, R_r^{a_r})$. A first-order interpretation (or \FO-interpretation)
\[I: \struc[σ] \to \struc[τ]\]
is given by a tuple of \FO-formulae $φ_0, φ_1, \dots, φ_r$ over vocabulary σ. For some $k$, $φ_0$ has $k$ free variables and $φ_i$ has $k \cdot a_i$ free variables for all $i \geq 1$. For each structure $\mathcal{A} \in \struc[σ]$, these formulae define the structure
\[I(\mathcal{A}) = \langle |I(\mathcal{A})|, R_1^{I(\mathcal{A})}, \dots, R_r^{I(\mathcal{A})} \rangle \in \struc[τ],\]
where the universe is defined by $φ_0$ and the relations by $φ_1, \dots, φ_r$ in the following way:
\[|I(\mathcal{A})| = \{\langle b^1, \dots, b^k \rangle \mid \mathcal{A} \vDash φ_0(b^1, \dots, b^k)\}\]
\[R_i^{I(\mathcal{A})} = \{(\langle b_1^1, \dots, b_1^k \rangle, \dots, \langle b_{a_i}^1, \dots, b_{a_i}^k \rangle) \in |I(\mathcal{A})|^{a_i} \mid \mathcal{A} \vDash φ_i(b_1^1, \dots, b_{a_i}^k)\}\]
\end{definition}
For better readability, we will write $φ_\text{universe}$ instead of $φ_0$ and $φ_{R_i}$ instead of $φ_i$ \fa $i$.

We will next recall the definition of Boolean circuits and complexity classes defined using them. A circuit is a directed acyclic graph (dag), whose nodes (also called gates) are marked with either a Boolean function (in our case $\land$ or $\lor$), a constant (0 or 1), or a (possibly negated) bit of the input. Also, one gate is marked as the output gate. On any input, a circuit computes a Boolean function by evaluating all gates according to what they are marked with. The value of the output gate then is the function value for that input. If $C$ is a circuit, we denote the function it computes by $C(x)$.\\
When we want circuits to work on different input lengths, we have to consider families of circuits: A family contains a circuit for any input length $n \in \mathbb{N}$. Families of circuits allow us to talk about languages being accepted by circuits: A circuit family $\mathcal{C} = (C_n)_{n \in \mathbb{N}}$ is said to accept (or decide) the language $L$, if it computes its characteristic function $c_L$:
\[C_{|x|}(x) = c_L(x) \textrm{ for all } x.\]
The complexity classes in circuit complexity are classes of languages that can be decided by circuit families with certain restrictions to their resources. The resources relevant here are depth, size and fan-in (number of children) of gates. The depth here is the length of a longest path from any input gate to the output gate of a circuit and the size is the number of non-input gates in a circuit. Depth and size of a circuit family are defined as functions accordingly.\\
Above, we have not restricted the computability of the circuit $C_{|x|}$ from $x$ in any way. This is called non-uniformity, which allows such circuit families to even compute non-recursive functions. Since we want to capture a kind of efficient computability, we need some notion of uniformity. For this, we first define the vocabulary for Boolean circuits as \FO-structures:
\[\tCirc = (E^2, G_\land^1, G_\lor^1, \text{Input}^2, \text{negatedInput}^2, \text{output}^1),\]
where the relations are interpreted as follows:
\begin{itemize}
  \item $E(x,y)$: gate $y$ is a child of gate $x$
  \item $G_\land(x)$: gate $x$ is an and-gate
  \item $G_\lor(x)$: gate $x$ is an or-gate
  \item $\text{Input}(x, i)$: gate $x$ is an input gate associated with the $i$'th input bit
  \item $\text{negatedInput}(x,i)$: gate $x$ is a negated input gate associated with the $i$'th input bit
  \item $\text{output}(x)$: gate $x$ is the output gate
\end{itemize}

We will now define \FO-uniformity of Boolean circuits and the complexity classes relevant in this paper.

\begin{definition}
A circuit family $\mathcal{C} = (C_n)_{n \in \mathbb{N}}$ is said to be first-order uniform (FO-uniform) if there is an FO-interpretation
\[I: \struc[\tString\cup(\textrm{BIT}^2)] \rightarrow \struc[\tCirc]\]
mapping any structure $\mathcal{A}_w$ over \tString with built-in BIT to the circuit $C_{|w|}$ given as a structure over the vocabulary \tCirc. 
\end{definition}

Note that by \cite{baim94} this uniformity coincides with the maybe better known DLOGTIME-uniformity for many familiar circuit classes (and in particular for all classes studied in this paper). All circuit classes we consider in this paper are \FO-uniform.

\begin{definition}
A language $L \subseteq \{0,1\}^*$ is in $\classFont{AC}^i$ if there is an \FO-uniform circuit family with depth $(\log n)^i$ and polynomial size accepting $L$. $\classFont{NC}^i$ is defined analogously with bounded fan-in gates. $\classFont{SAC}^i$ is defined analogously with bounded fan-in ∧-gates and unbounded fan-in ∨-gates.
\end{definition}
We will also call circuit families with the above restrictions on their resources $\classFont{AC}^i$, $\classFont{NC}^i$ and $\classFont{SAC}^i$ circuit families, respectively.

For this paper, the classes \AC, \ACone, \NC and \SAC are of particular interest. It is known that the class \AC coincides with the class $\FO$ \cite{baimst90,ImmermanBuch}: $\AC = \FO$. 

We will next define counting variants of the above classes. The idea for counting classes in general is to use a model of computation and identify a kind of witness for acceptance in that model. For a nondeterministic Turing machine, we usually consider the accepting paths on a given input as witnesses. Considering polynomial time computations, this concept gives rise to the class \cP.
A witness that a Boolean circuit accepts its input is a so-called \emph{proof tree}: a minimal subtree of the circuit showing that it evaluates to true for a given input. For this, we first unfold the given circuit into tree shape, and we further require that it is in \emph{negation normal form} (meaning that negations only occur directly in front of literals)---note that this is always the case for \tCirc-structures, though. A proof tree then is a subtree that contains the output gate, for every included $\lor$-gate exactly one child and for every included $\land$-gate all children, \ST every input gate which we reach in this way is a true literal. This allows us to define the following counting complexity classes:

\begin{definition}
A function $f\colon \{0,1\}^* \rightarrow \mathbb{N}$ is in $\#\classFont{AC}^i$ ($\#\classFont{NC}^i$, $\#\classFont{SAC}^i$) if there is an $\classFont{AC}^i$ ($\classFont{NC}^i$, $\classFont{SAC}^i$) circuit family $\mathcal{C} = (C_n)_{n \in \mathbb{N}}$ such that for all $w \in \{0,1\}^*$,
\[f(w) = \textrm{number of proof trees of } C_{|w|}(w).\]
\end{definition}

Note that, already at the first level, the classes $\#\classFont{AC}^1$, $\#\classFont{NC}^1$, $\#\classFont{SAC}^1$, though based on relatively close circuit classes,  have a rather different computational power. It can be seen, through the connections between $\classFont{SAC}^1$ circuits and multiplicatively disjoint circuits (see~\cite{Malod:2008ub}) that  $\#\classFont{NC}^1$ and $\#\classFont{SAC}^1$ are subclasses of $\#\classFont{P}$. On the contrary, the class $\#\classFont{AC}^1$ can output numbers bigger than $2^{n^{\log n}}$ for input of size $n$, hence numbers of super-polynomial sizes in $n$. This comes from the fact that the unfolding of a polynomial size, logarithmic depth circuit with  unbounded fan-in may be of size $n^{O(\log n)}$. This means that $\#\classFont{AC}^1\not\subseteq \#\classFont{P}$.

Similarly, one can identify witnesses for acceptance in first-order logic. One possibility for this is to do this in terms of the model-checking game defined as follows. The model checking game for \FO is the two-player game with the players ``verifier'' and ``falsifier'' played recursively on any \FO formula φ and input structure $\mathcal{A}$. The verifier tries to reach an atom that is satisfied while the falsifier tries to reach an atom that is not satisfied. For this, the game starts on the whole formula. From there, depending on the outermost operator or quantifier, one of the players makes a choice and the game continues on a certain sub-formula. The rules for this are as follows:
\begin{itemize}
  \item $\exists x ψ$: verifier chooses a value for $x$, continue on ψ
  \item $\forall x ψ$: falsifier chooses a value for $x$, continue on ψ
  \item $α \lor β$: verifier chooses whether to continue with $α$ or $β$
  \item $α \land β$: falsifier chooses whether to continue with $α$ or $β$
  \item $\neg α$: verifier and falsifier swap roles, continue on $α$
  \item For any atom : verifier wins if it is true, falsifier wins otherwise
\end{itemize}
In this game the verifier has a winning strategy---that is, a strategy that lets him win the game independent of the choices of the opponent---if and only if $\mathcal{A} \vDash φ$. This means that winning strategies in this game can be seen as witnesses for acceptance in first-order logic, which allows us to define a counting class based on \FO. 

\begin{definition} 
A function $f\colon \{0,1\}^* → ℕ$ is in \WinFO, if there is an \FO-formula φ over vocabulary $\tString \cup (\text{BIT}^2)$ \stfa $w \in \{0,1\}^*$:
\[f(w) = \text{CWin}(φ, \mathcal{A}_w),\]
where $\text{CWin}(φ, \mathcal{A}_w)$ is the number of winning strategies of the verifier in the model checking game for $\mathcal{A}_w \vDash φ$.
\end{definition}

As was shown in~\cite{HaakV16}, the counting versions of \AC and \FO coincide, \ie: $\cAC = \WinFO$.

For the quantifiers $∃_\text{b}$ and $∀_\text{b}$ we define the following rules in the model checking game for \FO: Here, the choosing player is restricted to the maximal two elements satisfying the relativization.

\section{GPR}\label{sect:GPR}

	We aim to characterize counting classes from circuit complexity beyond \cAC by counting winning strategies in different logics. It has been proved in~\cite{ComptonL90} that \NC can be characterized using \FO with a certain kind of linear recursion, called relational primitive recursion (RPR). It allows the recursive definition of predicates in the following way:
	\[[P(\tu x, y) \equiv θ(\tu x, y, P(\tu x, y-1))]\]
	\noindent where intuitively, $P(\tu x,y)$ has the same truth value as $θ(\tu x, y, P(\tu x, y-1))$ for $y>0$ and $P(\tu x,0)$ being equivalent to $θ(\tu x, y, \bot)$. Then, $\FO(\classFont{RPR})$ denotes the class of languages definable by formula of the form:
	
	\[[P(\tu x, y) \equiv θ(\tu x, y, P(\tu x, y-1))] \; \varphi(P)\]
	
	\noindent where $\varphi(P)$ is first-order and make use of the inductively defined $P$.
	Over structures with built-in BIT, it holds that $\NC = \FO(\classFont{RPR})$ \cite{ComptonL90}. 
	This characterization does not immediately generalize to classes  \SAC and \ACone as well as counting classes. However, inspired by this,  we define a different kind of inductive definition called \textbf{g}uarded \textbf{p}redicative \textbf{r}ecursion, \gpr for short, that allows us to capture all these classes in a unified way.

	\begin{definition}[\gpr]\label{def:gpr}
		 A formula $\phi\in \FO(\gpr)$ if it is of the form:
		
		\[\varphi ::= [P(\tu x, \tu y) \equiv \theta(\tu x,\tu y, P)] \; \varphi(P) \; | \; \psi\]

		\noindent where $\psi\in \FO$  and $\theta\in \FO$  with free variables $\tu x, \tu y$ such that each  atomic sub-formula involving symbol $P$
		
		\begin{enumerate}
			\item\label{cond1}  is of the form $P(\tu x,\tu z)$ where $\tu z$ is in the scope of a guarded quantification $Q\tu z.(\tu z\leq \tu y/2)$ with $Q\in \{\forall, \exists\}$ and
			\item\label{cond2} never occur in the scope of any quantification not guarded in this way.
		\end{enumerate}

    We also call the part in $[\makebox[1ex]{\textbf{•}}]$ a \gpr-operator.	We define  $\FO(\gpr_\bounded)$ similary by allowing only bounded variants for guarded quantification $Q_b \tu z.(\tu z < \tu y / 2)$ and $\FO(\gpr_\semi)$ for which universal guarded  $\forall\tu z.(\tu z\leq \tu y/2)$ and  bounded existential guarded $\exists_b \tu z.(\tu z < \tu y / 2)$  quantifications are allowed.
		\end{definition}
		
This approach is flexible enough to easily  express problems computable by small circuit classes.

\begin{example}
	The \textsc{shortcake} problem, proved $\ACone$-complete in \cite{ChandraTompa:1988}  is defined as follows. Two players, $H$ (or $0$) and $V$ (or $1$)  are alternately moving a token on an $n\times n$ Boolean matrix $M$. A configuration of the game is a contiguous submatrix of $M$ given that the indices of its first and last lines and columns $(i_0,j_0,i_1,j_1)$. It is given, at each round, with an indication of  which corner of this submatrix the token is on  and whose turn it is.  In the beginning of the game, the configuration is thus $(1,n,1,n)$, the token is at the $(1,1)$ corner and $H$ starts to play.  In his turn, $H$  tries to move the token horizontally in the  submatrix to some entry~\footnote{Supposing $j$ different from both $j_0,j_1$ allows to forget the precise corner where the token is and simplify in a non essential way the formula} $(i_0,j)$, $j\neq j_0,j_1$ satisfying $M_{i_0,j} = 1$. After, $H$'s move  either all columns to the left of $j$ or all columns to the right of $j$ are removed from the current submatrix, whichever number of columns is greater, leaving the token once again on a corner of the current submatrix. I.e.  the new configuration is $(i_0,j,i_1,j_1)$ if $j-j_0\leq j_1-j$ and $(i_0,j_0,i_1,j)$ if not.  In his turn,  $F$ plays similarly but vertically, on the rows. The first player with no move left loses.
	
	We encode the matrix by a structure representing a binary word of length $n^2$. Remark, that the size of the matrix is divided by at least two after each round. 
	The existence of a winning strategy for $H$ is encoded by the following $\FO(\gpr)$ formula ($s$ is an upper bound for the size of the matrix at each round with some padding, $p=0,1$ is for the players), 
	
	\[\phi \defeq [P(x,\underbrace{s, i_0,i_1,j_0,j_1,p}_{\tu y}) \equiv (p=0\wedge \theta_H(x,\tu y)) \vee (p=1\wedge \theta_V(x,\tu y))] \; \varphi(P) \]

	with $\varphi(P) \equiv P(n,2n^2,1,n,1,n,0)$ and $\theta_H(x, \tu y)$ is
	
	\[
\begin{array}{l}
	\exists \tu z=(s',i_0',i_1',j_0',j_1',p').(\tu z<  \tu y/2)\\ 
\qquad \left\{	\begin{array}{l}
(s'\leq s/2-2) \wedge  p'=1 \wedge i'_0=i_0 \wedge i'_1=i_1 \wedge P(x,s,i_0',j_0',i_1',j_1',p')) \wedge  \\
\big[ (M(i_0'n+j_0') \wedge j_0'\neq j_0,j_1 \wedge (j_1 - j_0)/2 \leq j_0' \leq j_1 \wedge j'_1=j_1)  \vee\\
(M(i_0'n+j_1') \wedge j_1'\neq j_0,j_1 \wedge j_0\leq j_1' \leq (j_1 - j_0)/2  \wedge j'_0=j_0) )\big]
\end{array}
\right.
\end{array}
	\]
	
	The formula $\theta_V(x,\tu y)$ associated to $V$ is defined similarly (but with universally guarded  quantification for $z$ and (row) $i$). In~\cite{ChandraTompa:1988} a variant of this game, called~\textsc{semicake} is shown to be $\SAC$-complete: it is easily definable along the same lines in $\FO(\gpr_\semi)$.
\end{example}





We now introduce a certain normal-form for circuits showing membership in \NC, \SAC and \ACone, which will be needed for our later proofs. Note that due to built-in BIT, we have an order and arithmetic on the gates of circuits from uniform circuit families. Circuit families in our normal-form have the following properties: All tuples of the appropriate size are gates (so $φ_\text{universe}$ from the \FO-interpretation showing uniformity is always true). The ∧-gates are exactly the gates that are odd and neither input nor negated input gates. The ∨-gates are exactly the gates that are even and neither input nor negated input gates. Children of gates are smaller than half of each of their parents. 

\begin{lemma}\label{lem:circuitNormalform}
Let $\mathfrak{C} \in \{\NC, \SAC, \ACone, \#\NC, \#\SAC, \#\ACone\}$ and $L \in \mathfrak{C}$. Then there is an \FO-interpretation $I: \struc[\tString] \cup (\textnormal{BIT}^2) \to \struc[\tCirc]$ with tuple size $k \in \mathbb{N}$ that uniformly describes a circuit family showing $L \in \mathfrak{C}$ \stfa $w \in \{0,1\}^*$:
\begin{enumerate}
  \item $|I(\mathcal{A}_w)| = |\mathcal{A}|^k$
  \item \fa $\tu x \in |I(\mathcal{A}_w)|$:\\
  $G_\land^{I(\mathcal{A}_w)}(\tu x) \Leftrightarrow (\neg \textnormal{Input}^{I(\mathcal{A}_w)}(\tu x) \land \neg \textnormal{negatedInput}^{I(\mathcal{A}_w)}(\tu x) \land \tu x \textnormal{ is odd})$\\
  $G_\lor^{I(\mathcal{A}_w)}(\tu x) \Leftrightarrow (\neg \textnormal{Input}^{I(\mathcal{A}_w)}(\tu x) \land \neg \textnormal{negatedInput}^{I(\mathcal{A}_w)}(\tu x) \land \tu x \textnormal{ is even})$

  \item \fa $\tu x, \tu y \in |I(\mathcal{A}_w)|$:\\
  $E^{I(\mathcal{A}_w)}(\tu x, \tu y) ⇒ ∃\tu y' 2 • \tu y = \tu y' ∧ \tu y' < \tu x$ 
\end{enumerate}
\end{lemma}

\begin{proof}
Properties 1 and 2 are straightforward. For property 3, a certain unary encoding of the depth can be added to the encoding of gates in order to halve the numerical value of gates in each step from parent to child.

A formal proof can be found in the appendix.

%
%
\end{proof}

\section{Logical Characterizations of Small Depth Decision Classes}\label{sect:dec}

We now show that the newly defined logics characterize the classes \NC, \SAC and \ACone, respectively.

\begin{theorem}\label{thm:mainDecision}
\leavevmode
\begin{enumerate}
  \item $\NC = \FO(\gpr_\bounded)$
  \item $\SAC = \FO(\gpr_\semi)$
  \item $\ACone = \FO(\gpr)$
\end{enumerate}
\end{theorem}

\begin{proof}
\underline{$\ACone \subseteq \FO(\gpr)$}: Let $L \in \ACone$ via the \FO-uniform \ACone circuit family $\mathcal{C} = (C_n)_{n \in \mathbb{N}}$ with the properties from Lemma \ref{lem:circuitNormalform} and $C_n$ has depth at least 1 for all $n$. The latter can easily be achieved by adding a new ∧-gate as output-gate with the old output-gate being its only child. Let
\[I = (φ_\textnormal{universe}, φ_{G_\land}, φ_{G_\lor}, φ_\textnormal{Input}, φ_\textnormal{negatedInput}, φ_E, φ_\textnormal{output})\]
be an \FO-interpretation showing that $\mathcal{C}$ is uniform. %
Furthermore, let
\[φ_\textnormal{Literal}(\tu x) \defeq \exists \tu i (φ_\textnormal{Input}(\tu x, \tu i) \lor φ_\textnormal{negatedInput}(\tu x, \tu i)),\]
\[φ_\textnormal{trueLiteral}(\tu x) \defeq \exists \tu i (φ_\textnormal{Input}(\tu x, \tu i) \land S(\tu i) \lor φ_\textnormal{negatedInput}(\tu x, \tu i) \land \neg S(\tu i)) \text{ and}\]
\[ψ(\tu z, P(\tu z)) = P(\tu z) ∧ ¬φ_\text{Literal}(\tu z) ∨ φ_\text{trueLiteral}(\tu z).\]
Then the following $\FO(\gpr)$-formula defines $L$:
\[Φ \defeq [P(\tu y) \equiv θ(\tu y, P)] \; ∃\tu o (φ_\text{output}(\tu o) ∧ P(\tu o))\]
with
\begin{align*}
θ(\tu y, P) &\defeq \Big(\text{Even}(\tu y) ∧ \big((∃\tu z.(\tu z < \tu y/2 ∧ φ_E(\tu y, \tu z))) \; ψ(\tu z, P(\tu x, \tu z))\big)\Big) ∨\\
&\phantom{\defeq \ \,} \Big(\text{Odd}(\tu y) ∧ \big((∀\tu z.(\tu z < \tu y/2 ∧ φ_E(\tu y, \tu z))) \; ψ(\tu z, P(\tu x, \tu z))\big)\Big).
\end{align*}
Even and Odd check the parity of the least significant bit within the least significant variable within tuple $\tu y$ using BIT. Note that $φ_E(\tu y, \tu z)$ within the relativization for $\tu z$ can be moved outside the relativization, so Φ is equivalent to a $\FO(\gpr)$-formula.

Since $\text{Odd}(\tu y) \equiv ¬\text{Odd}(\tu y)$, we can write θ as
\[θ(\tu y, P) \equiv \big(Q\tu z.(\tu z < \tu y/2 ∧ φ_E(\tu y, \tu z))\big) \; P(\tu z) ∧ ¬φ_\text{Literal}(\tu z) ∨ φ_\text{trueLiteral}(\tu z)\]
where $Q$ is either ∃ or ∀ depending on the parity of $\tu y$.

Let $n \in \mathbb{N}$ and $w \in \{0,1\}^n$. We now prove that the predicate $P$ in the above formula is the valuation for the gates in circuit $C_n$. By definition, on input structure $\mathcal{A}_w$, the formulae from $I$ used above give access to $C_n$. We prove inductively that for any $k \in \mathbb{N}$, $P(\tu g)$ gives the value of gate $\tu g$ in $C_n$ on input $w$ if all children of $\tu g$ have depth $\leq k$.

$k = 0$: Note that $φ_\textnormal{trueLiteral}(\tu h)$ gives the value of $\tu h$ in $C_n$ on input $w$ if $\tu h$ is an input gate. Then for gates $\tu g$ all children of which are input gates we have:
\begin{equation*}\label{eq:IAdec}
P(\tu g) \equiv \big(Q\tu z.(\tu z < \tu g/2 \land φ_E(\tu g,\tu z))\big) \; \big(\underbrace{P(\tu z) \land \underbrace{\neg φ_\textnormal{Literal}(\tu z)}_\textnormal{false}}_\textnormal{false} \lor φ_\textnormal{trueLiteral}(\tu z)\big).\tag{$\star$}
\end{equation*}
By assumption, if $φ_E(\tu g, \tu z)$ then $\tu z < \tu g/2$, and thus
\[\tu z < \tu g/2 \land φ_E(\tu g, \tu z) \equiv φ_E(\tu g, \tu z).\]
This yields
\[P(\tu g) \equiv \big(Q\tu z.φ_E(\tu g, \tu z)\big) \; φ_\textnormal{trueLiteral}(\tu z)\]
This means that $P$ actually gives the value of $\tu g$.

$k \to k+1$: Again, by assumption,
\[\tu z < \tu g/2 \land φ_E(\tu g, \tu z) \equiv φ_E(\tu g, \tu z).\]
We also know that for all children $\tu z$ of $\tu g$ only two cases can occur:

If $\tu z$ is an input gate, then $\neg φ_\textnormal{Literal}(\tu z)$ is false and $φ_\textnormal{trueLiteral}(\tu z)$ gives the value of $\tu z$.

If $\tu z$ is not an input gate, then $φ_\textnormal{trueLiteral}(\tu z)$ is false, $¬φ_\textnormal{Literal}$ is true and $P(\tu z)$ gives the value of $\tu z$ by induction hypothesis.

By (\ref{eq:IAdec}) this means that $P(\tu g)$ actually gives the value of $\tu g$.

Since $P$ gives the value of arbitrary non-input gates in $C_n$ on input $w$ for any $n$ and $w$ and we assumed that the output gate is not an input gate, it is easy to see that the above formula defines $L$: The formula behind the recursive definition of $P$ simply states that the output gate of the circuit evaluates to true.

\underline{$\FO(\gpr) \subseteq \ACone$}:
At first assume that only one occurrence of \gpr-operators is allowed. The proof easily extends to the general case. Furthermore, we begin by proving the result without negations in θ. We will explain how to handle arbitrary $\FO(\gpr)$-formulae afterwards.

Let $L \in \FO(\gpr)$ via the formula
\[[P(\tu x, \tu y) \equiv θ(\tu x, \tu y, P)] \; φ(P).\]
By definition of $\FO(\gpr)$, $P$ occurs in θ only in the form $P(\tu x, \tu z)$, where $\tu z$ is in the scope of a guarded quantification $Q\tu z.(\tu z < \tu y/2)$ with $Q \in \{∃,∀\}$ and not in the scope of any unguarded quantification.

Ignoring occurrences of $P$, φ is an \FO-formula. Hence, we can build an \AC circuit family evaluating φ except for these occurrences.

In order to compute the predicate $P$ we proceed as follows:

θ is also an \FO-formula except for occurrences of $P$, so we can build for all $\tu x, \tu y$ a \AC circuit that computes $θ(\tu x, \tu y, P)$ with certain input gates marked with $P(\tu x, \tu z)$. The circuit can easily be built in a way that $\tu z$ is part of the encoding of gates that are marked with $P(\tu x, \tu z)$. Thus, we can remove the marks and instead connect each gate that was marked with $P(\tu x, \tu z)$ to the output gate of the subcircuit computing $P(\tu x, \tu z)$. Since occurrences of $P(\tu x, \tu z)$ only occur within guarded quantifications $Q\tu z.(\tu z < \tu y/2)$, there can be at most logarithmically many steps from any $P(\tu x, \tu y)$ before reaching $P(\tu x, \tu 0)$, terminating the recursion. By the above, each such step---computing $P(\tu x, \tu y)$, when given values of $P(\tu x, \tu z)$ for certain $\tu z$---can be done in constant depth leading to logarithmic depth in total. 

The gates computing values of $P$ can now be connected to the \AC circuit family evaluating φ as needed. This leads to an \ACone circuit family evaluating the whole formula.

Next, we talk about the case of θ containing negations. For this, we use the same construction as above, but add a negated version of each gate. We do this by adding a negation-bit to the encoding of all gates (possibly with padding). This is toggled exactly when negations occur in the quantifier-free part.
For example, consider a subformula $α = β ∧ \neg γ$ and assume there was no negation around α. Then we have a gate $g$, which will compute the truth-value of α, for which the negation-bit is 0. We connect this to the gate for the truth-value of β with negation-bit 0 and---since there is a negation around γ---the gate for the truth-value of γ with negation-bit 1.

Apart from constructing the connections in this way, the negation-bit also changes the gate-type of gates: If a non-negated gate is a ∨-gate, the negated version is a ∧-gate and vice versa. Also, negated gates computing the value of literals also use the negated version of the respective literal compared to the non-negated version.

In total, this construction only doubles the size of the circuit and does not change its depth, but handles arbitrary negations.

For the case of multiple \gpr-operators, we build a circuit for each of them in the above way. In case of nesting, we start from the innermost operator. Adequate connections between the different circuits are easily doable and size and depth of the combination of all those circuits still stays within the desired bounds.

\underline{$\NC \subseteq \FO(\gpr_\bounded)$ and $\SAC \subseteq \FO(\gpr_\semi)$}: Can be shown with the same formula and the same proof as $\ACone \subseteq \FO(\gpr)$, replacing $\gpr$ by $\gpr_\bounded$ or $\gpr_\semi$, respectively.

\underline{$\FO(\gpr_\bounded) \subseteq \NC$}: This can be proven completely analogously to $\FO(\gpr) \subseteq \ACone$. Instead of \AC circuit families for evaluation of φ and θ, we now use \NC circuit families. This leads to logarithmic depth for evaluation of θ. In general, repeating this for logarithmically many steps would be a problem. By definition there are no occurrences of $P$ inside any unbounded quantifier, though. For the bounded quantifiers, we still create subcircuits for all possible values for the quantified variable, but we only connect the maximal two satisfying the relativization to the parent. This ensures that gates marked with $P(\tu x, \tu z)$ for some $\tu x, \tu z$ still only occur in constant depth in the circuit evaluating $θ(\tu x, \tu y, P)$, this time with only bounded fan-in gates. Consequently, the construction still only leads to logarithmic depth in total.

\underline{$\FO(\gpr_\semi) \subseteq \SAC$}: Here, the same trick as for \NC can be used. θ can be evaluated using an \NC circuit family which is also an \SAC circuit family. Also, the semi-unboundedness of the quantifiers around occurrences of $P$ directly corresponds to the semi-unboundedness in \SAC circuit families.

\end{proof}

The proof of the inclusion $\ACone \subseteq \FO(\gpr)$ also immediately gives us the following normal-form for our logical classes.
\begin{corollary}
Let $\mathcal{G} \in \{\gpr, \gpr_\bounded, \gpr_\semi\}$. Then
\[\FO(\mathcal{G}) = \mathcal{G}\classFont{-FO},\]
where $\mathcal{G}\classFont{-FO}$ denotes the class of languages decidable in first-order logic with one \gpr-operator in the beginning.
\end{corollary}

\section{Logical Characterizations of Small Depth Counting Classes}\label{sect:count}

Next, we want to define a game semantics for our new logics. The game we define will correspond to model-checking and is defined analogous to the model checking game for \FO for the most part. When playing the game on an $\FO(\gpr)$-formula
\[[P(\tu x, \tu y) \equiv θ(\tu x, \tu y, P] \; φ(P),\]
it begins on formula φ. The only difference to the model checking for \FO is an additional case for atoms of the form $P(\tu a, \tu b)$. In this case, the game continues on the formula $θ(\tu x, \tu y, P)$. For all other atoms, the winner is immediately determined as before.

Now for any $\mathcal{A}$ and $φ \in \FO(\gpr)$ it holds that
\[\mathcal{A} \vDash φ \Longleftrightarrow \textnormal{the verifier has a winning strategy for the game on } \mathcal{A} \vDash φ.\]
Analogously, we can extend the semantic game for $\FO(\gpr_\bounded)$ and $\FO(\gpr_\semi)$.

Similar to the approach in \cite{HaakV16}, we can also count the number of winning strategies of the verifier. 

\begin{definition}
A function $f \colon \{0,1\}^* → ℕ$ is in $\WinFO(\gpr)$, if there is an $\FO(\gpr)$-formula φ over vocabulary $\tString \cup (\text{BIT}^2)$ \stfa $w \in \{0,1\}^*$:
\[f(w) = \text{CWin}(φ, \mathcal{A}_w),\]
where $\text{CWin}(φ, \mathcal{A}_w)$ is the number of winning strategies of the verifier in the model checking game for $\mathcal{A}_w \vDash φ$.
\end{definition}
$\WinFO(\gpr_\bounded)$ and $\WinFO(\gpr_\semi)$ are defined analogously.

This then gives us characterizations of the counting version of the corresponding classes from circuit complexity:

\begin{theorem}\label{thm:mainCounting}
\leavevmode
\begin{enumerate}
  \item $\#\NC = \WinFO(\gpr_\bounded)$
  \item $\#\SAC = \WinFO(\gpr_\semi)$
  \item$\#\ACone = \WinFO(\gpr)$
\end{enumerate}
\end{theorem}

\begin{proof}
The proof consists of carefully counting winning strategies in semantic games for those formulae developed in the decision version (Theorem~\ref{thm:mainDecision}) and
is given in the appendix.
\end{proof}

Analogously to the decision version, the proof again allows us to establish a normal-form for our new logical classes.
\begin{corollary}
Let $\mathcal{G} \in \{\gpr, \gpr_\bounded, \gpr_\semi\}$. Then
\[\WinFO(\mathcal{G}) = \#\classFont{Win-}\mathcal{G}\classFont{-FO},\]
where $\#\classFont{Win-}\mathcal{G}\classFont{-FO}$ denotes the class of functions that can be described as the number of winning strategies for first-order formulae with one \gpr-operator in the beginning.
\end{corollary}

\begin{remark}
To further show the robustness of our classes, we want to mention certain variations of our logics that do not change the resulting complexity classes. For all decision classes, we can drop condition \ref{cond2} from Definition \ref{def:gpr} without changing the class. For $\WinFO(\gpr)$ the same holds.

For $\WinFO(\gpr_\bounded)$ and $\WinFO(\gpr_\semi)$, condition 2 cannot be dropped but can be replaced by the following weaker version: 
  ``never occur in the scope of any universal quantification not guarded in this way''.
\end{remark}

\section{Conclusion}\label{sect:concl}

We extended the only so-far known logical characterization of an arithmetic circuit class, namely $\cAC = \WinFO$ \cite{HaakV16}, to arithmetic classes defined by circuits of logarithmic depth.
In order to achieve this, we first had to develop logical characterizations of the corresponding Boolean classes.

The result from \cite{HaakV16} was used in \cite{durand_et_al:LIPIcs:2016:6560} to place $\cAC$ in a strict hierarchy of counting classes within $\cP$. In this way, lower bounds for several logically-defined arithmetic classes were obtained. 
Our hope is that the here presented characterizations of larger arithmetic classes will also lead to new insights about these and hopefully spur development of new upper and lower bounds, \eg, is $\#\NC\subseteq\NC$? Is $\#\NC\neq\cP$? Is $\NC\neq\PP$?

\bibliographystyle{plain}
\bibliography{references}

\newpage
\appendix

\section{Appendix}

\begin{proof}[Proof of Lemma \ref{lem:circuitNormalform}]
Since $L \in \mathfrak{C}$, there is an \FO-interpretation $I: \struc[\tString] \to \struc[\tCirc]$ that uniformly describes a circuit family showing $L \in \mathfrak{C}$. Let $k$ be the length of tuples encoding gates in this circuit family. We now stepwise construct an \FO-interpretation still showing $L \in \mathfrak{C}$ but with properties 1, 2 and 3.

Since we will have to adapt encodings of gates as tuples in order to manipulate certain properties related to the numerical predicates, it is relevant for this proof in what way we represent numbers as tuples. We will always use an \textit{most significant bit first} encoding, meaning that the variable containing the most significant bit of number is the left-most variable in a tuple and significance reduces towards the left.

\underline{1.}: We first construct an \FO-interpretation $I'$ that has property 1. This is done by allowing all tuples as gates, but allowing connections between gates only if both gates were already gates in the original circuit. Analogously, we also have to change the formula determining the output gate---otherwise, multiple tuples could become output gates. The only formulae we have to change for this are those for the universe and the predicates $E$ and $\textnormal{output}$. Let $φ_\textnormal{universe}, φ_E, φ_{\textnormal{output}}$ be the respective formulae from $I$. For $I'$ we instead use
\begin{align*}
φ_\textnormal{universe}'(\tu g) &= \top\\
φ_E'(\tu g_1, \tu g_2) &= φ_E(\tu g_1, \tu g_2) \land φ_\textnormal{universe}(\tu g_1) \land φ_\textnormal{universe}(\tu g_2)\\
φ_\textnormal{output}'(\tu g) &= φ_\textnormal{output}(\tu g) \land φ_\textnormal{universe}(\tu g)
\end{align*}

Now, in order to additionally achieve properties 2 and 3 from the statement of the lemma, we proceed as follows: For property 2, we add an additional bit as LSB to the encoding of gates and use only those versions of ∧-gates where this bit is 1 and those versions of ∨-gates where this bit is 0. For property 3, we add to the encoding of gates a unary encoding of the depth in a certain way. We then only connect two gates if their corresponding gates from the original circuit were connected and the depth increases by 1 from the parent to the child. (This means a lot of these new versions of the gates are not used.) 

For both these approaches we want to add bits to the encoding. Since we can only add additional variables to the tuples encoding gates and each variable increases the number of bits by $\log n$, we pad the number of bits to a multiple of $\log n$ and simply add the according number of additional variables to the tuples.

We now formalize the above ideas.

\underline{2.}: We construct an \FO-interpretation $I''$ that has property 2 in addition to the previous properties. For this, we increase the tuple size by 1. Then the following formulae can be used for $I''$:
\begin{align*}
φ_\textnormal{universe}''(\tu gx) &= \top\\
φ_{G_\land}''(\tu gx) &= \textnormal{BIT}(x,0)\\
φ_{G_\lor}''(\tu gx) &= \neg \textnormal{BIT}(x,0)\\
φ_\textnormal{Input}''(\tu gx, i_1 \dots i_k y) &= i_1 = 0 \land φ_\textnormal{Input}'(\tu g, i_2 \dots i_k y)\\
φ_\textnormal{negatedInput}''(\tu gx, i_1 \dots i_k y) &= i_1 = 0 \land φ_\textnormal{negatedInput}'(\tu g, i_2 \dots i_k y)\\
φ_E''(\tu g_1x_1, \tu g_2x_2) &= φ_E'(\tu g_1, \tu g_2) \land \bigwedge_{i=1}^2 ψ_\textnormal{real}(\tu g_ix_i)\\
φ_\textnormal{output}''(\tu gx) &= φ_\textnormal{output}'(\tu g) \land ψ_\textnormal{real}(\tu gx)
\end{align*}
with
\[ψ_\textnormal{real}(\tu gx) = (φ_{G_\land}'(\tu g) \land x=1) \lor (φ_{G_\lor}'(\tu g) \land x=0) \lor \exists \tu i (φ_\textnormal{Input}'(\tu g, \tu i) \lor φ_\textnormal{negatedInput}'(\tu g, \tu i)).\]

\underline{3.}: We construct an \FO-interpretation $I'''$ that additionally has property 3. For this, let $\ell \cdot \log n$ be a bound for the depth of the circuit family described by $I''$. For $I'''$, we increase the tuple size by $2\ell$. The idea is to create for each gate of the old circuit a duplicate for each possible depth within the circuit. This will make the circuit layered. The depth is encoded in the following way: Depth 0 is encoded by the sequence only consisting of 1s. Depth $i+1$ is then encoded by the same sequence as $i$ only the first two 1s are made 0s. This kind of unary encoding is possible since the circuits have logarithmic depth. In the final circuit, gates are only connected if their versions in the old circuit were connected and the child's depth is 1 higher than the parent's. This leads to the following formulae for $I'''$:

\begin{align*}
φ_\textnormal{universe}'''(\tu h \tu g x) &= \top\\
φ_{G_\land}'''(\tu h \tu g x) &= φ_{G_\land}''(\tu g x)\\
φ_{G_\lor}'''(\tu h \tu g x) &= φ_{G_\lor}''(\tu g x)\\
φ_\textnormal{Input}'''(\tu h \tu g x, \tu j \tu i y) &= \tu j = \tu 0 \land φ_\textnormal{Input}''(\tu g x, \tu i y)\\
φ_\textnormal{negatedInput}'''(\tu h \tu g x, \tu j \tu i y) &= \tu j = \tu 0 \land φ_\textnormal{negatedInput}''(\tu g x, \tu i y)\\
φ_E'''(\tu h_1 \tu g_1 x_1, \tu h_2 \tu g_2 x_2) &= φ_E''(\tu g_1 x_1, \tu g_2 x_2) \land \exists \tu i \bigg( \big(\forall \tu j \leq \tu i \textnormal{BIT}(\tu h_1, \tu j)\big) \land \big(\forall \tu j > \tu i \neg \textnormal{BIT}(\tu h_1, \tu j)\big)\\
& \quad \quad \quad \quad \big(\forall \tu j \leq (\tu i-2) \textnormal{BIT}(\tu h_2, \tu j)\big) \land \big(\forall \tu j > (\tu i -2) \neg \textnormal{BIT}(\tu h_2, \tu j)\big) \bigg)
\end{align*}

The circuit family described by $I'''$ has properties 1, 2 and 3.
\end{proof}

\begin{proof}[Proof of Theorem \ref{thm:mainCounting}]
The proof idea is very similar to the one used for the decision version. We again begin by proving the unbounded version.

\underline{$\#\ACone \subseteq \WinFO(\gpr)$}:
Let $f \in \#\ACone$ via the \FO-uniform \ACone circuit family $\mathcal{C} = (C_n)_{n \in ℕ}$ with the properties from Lemma \ref{lem:circuitNormalform} and $C_n$ has depth at least 1 for all $n$. The latter can easily be achieved by adding a new ∧-gate as output-gate with the old output-gate being its only child. Let
\[I = (φ_\textnormal{universe}, φ_{G_\land}, φ_{G_\lor}, φ_\textnormal{Input}, φ_\textnormal{negatedInput}, φ_E, φ_\textnormal{output})\]
be an \FO-interpretation showing that $\mathcal{C}$ is uniform. Furthermore, let θ and its subformulae ψ, $φ_\text{Literal}$ and $φ_\text{trueLiteral}$ be defined as in the proof of Theorem \ref{thm:mainDecision}. Then
\[Φ \defeq [P(\tu y) \equiv θ(\tu y, P)] \; ∃\tu o(φ_\text{output}(\tu o) ∧ P(\tu o))\]
defines $f$. Note that $φ_E(\tu y, \tu z)$ within the relativization for $\tu z$ can be moved outside the relativization without changing the number of winning strategies, leading to an $\FO(\gpr)$-formula with the same number of winning strategies.

In the following we use $\#\win(φ, \mathcal{A})$ as notation for the number of winning strategies of the verifier for $\mathcal{A} \vDash φ$. $\text{Odd}(\tu y)$ can be constructed \stfa $\mathcal{A}$ we have $\#\win(\text{Odd}(\tu y), \mathcal{A}) \in \{0,1\}$ and $\text{Even}(\tu y)$ can be constructed \stfa $\mathcal{A}$ we have $\#\win(\text{Even}(\tu y), \mathcal{A}) = 1 - \#\win(\text{Odd}(\tu y), \mathcal{A})$. This means that for all $\mathcal{A}$ it holds that
\[\#\win(θ(\tu y, P), \mathcal{A}) = \#\win\Big(\big(Q\tu z.(\tu z < \tu y/2 ∧ φ_E(\tu y, \tu z))\big) P(\tu z) ∧ ¬φ_\text{Literal}(\tu z) ∨ φ_\text{trueLiteral}(\tu z), \mathcal{A}\Big)\]
where $Q$ is either ∃ or ∀ depending on the parity of $\tu y$.

Let $n \in ℕ$ and $w \in \{0,1\}^n$. We now prove that the number of winning strategies verifying that $P(\tu g)$ is true is exactly the number of proof trees of the subcircuit of $C_n$ rooted in $\tu g$. By definition, on input structure $\mathcal{A}_w$, the formulae from $I$ used above give access to $C_n$. We prove inductively that for any $k \in ℕ$, $\#\win(P(\tu g))$ gives the number of proof trees of the subcircuit of $C_n$ rooted in $\tu g$ on input $w$ if all children of $\tu g$ have depth $\leq k$.

$k = 0$: Note that $φ_\textnormal{trueLiteral}(\tu h)$ gives the value of $\tu h$ in $C_n$ on input $w$ if $\tu h$ is an input gate and that for these gates the value of the gate is equal to the number of proof trees of the subcircuit rooted in them. This means that for gates $\tu g$ all children of which are input gates we have:
\begin{align*}\label{eq:IAcount}
\#\win(P(\tu g)) &= \#\win(Q\tu z.(\tu z < \tu g/2 ∧ φ_E(\tu g, \tu z)) \big( P(\tu z) ∧ ¬φ_\textnormal{Literal}(\tu z) ∨ φ_\textnormal{trueLiteral}(\tu z) \big)\\
&= \ \ \quad \mathclap{\mathop{\bigcirc}_{\substack{\tu z \in |\mathcal{A}_w|,\\ \tu z<\tu g/2 ∧ φ_E(\tu g, \tu z)}}} \quad \#\win( P(\tu z) ∧ ¬φ_\textnormal{Literal}(\tu z) ∨ φ_\textnormal{trueLiteral}(\tu z) ), \tag{$\star \star$}
\end{align*}
where $\bigcirc$ is either summation or multiplication depending on the parity of $\tu g$.

We can assume that there is exactly one winning strategy showing $φ_\textnormal{Literal}$, respectively $φ_\textnormal{trueLiteral}$, if it is true (and none otherwise). Since $k=0$, we know that for all $\tu z$ that meet the conditions, $φ_\textnormal{Literal}(\tu z)$ is true yielding
\[\#\win(P(\tu g)) = \ \ \quad \mathclap{\mathop{\bigcirc}_{\substack{\tu z \in \mathcal{A}_w,\\\tu z<\tu g/2 ∧ φ_E(\tu g, \tu z)}}} \quad \#\win(φ_\textnormal{trueLiteral}(\tu z)).\]
By assumption, if $φ_E(\tu g, \tu z)$ then $\tu z < \tu g/2$, and thus $\tu z < \tu g/2 ∧ φ_E(\tu g, \tu z) \equiv φ_E(\tu g, \tu z)$.

This means that $\#\win(P(\tu g))$ is exactly the number of proof trees of the subcircuit of $C_n$ rooted in $\tu g$.

$k → k+1$: Again, by assumption, $\tu z<\tu g/2 ∧ φ_E(\tu g, \tu z) \equiv φ_E(\tu g, \tu z)$. We also know that for all children $\tu z$ of $\tu g$ only two cases can occur:

If $\tu z$ is an input gate, then $\#\win(¬φ_\textnormal{Literal}(\tu z)) = 0$ and $\#\win(φ_\textnormal{trueLiteral}(\tu z))$ is exactly the number of proof trees of the subcircuit of $C_n$ rooted in $\tu z$.

If $\tu z$ is not an input gate, then by assumption $\#\win(φ_\textnormal{trueLiteral}(\tu z)) = 0$, $\#\win(¬φ_\textnormal{Literal}) = 1$ and by induction hypothesis $\#\win(P(\tu z))$ is exactly the number of proof trees of the subcircuit of $C_n$ rooted in $\tu z$.

By (\ref{eq:IAcount}) this means that $\#\win(P(\tu g))$ is equal to the number of proof trees of the subcircuit of $C_n$ rooted in $\tu g$.

Since $\#\win(P(\tu g))$ gives the number of proof trees of the subcircuit of $C_{|w|}$ rooted in $\tu g$ for arbitrary non-input gates $\tu g$ in $C_n$ on input $w$ for any $w$, it is easy to see that the above formula defines $f$: The number of winning strategies of the formula behind the recursive definition of $P$ is almost immediately the number of winning strategies for $P(\underline{\textnormal{output}})$, where \underline{output} is the unique element satisfying $φ_\text{output}$. By the induction above this is equal to the number of proof trees of circuit $C_n$.

\underline{$\WinFO(\gpr) \subseteq \#\ACone$}: This can be proven completely analogously to the decision version. Counting proof trees of the constructed circuit family leads exactly to the function given by the number of winning strategies of the formula we started with.

Both $\#\NC \subseteq \WinFO(\gpr_\bounded)$ and $\#\SAC \subseteq \WinFO(\gpr_\semi)$ can---as for the decision version---be shown with the same formula as $\#\ACone \subseteq \WinFO(\gpr)$ by changing the \gpr-operator to a $\gpr_\bounded$- or $\gpr_\semi$-operator, respectively.

The converse directions $\WinFO(\gpr_\bounded) \subseteq \#\NC$ and $\WinFO(\gpr_\semi) \subseteq \#\SAC$ can also be shown analogoulsy, using again the restriction that $P$ occurs only within bounded quantifiers within θ.
%
\end{proof}

\end{document}